%% file: root_arxiv.tex
\documentclass[12pt]{article}

\usepackage{amsmath,amssymb,amsfonts,mathrsfs,mathtools,dsfont,mathrsfs, amsthm, bm, bbm}
\DeclareMathAlphabet{\mathpzc}{OT1}{pzc}{m}{it}
\DeclareSymbolFontAlphabet{\amsmathbb}{AMSb}%

\usepackage[dvipsnames]{xcolor}
\usepackage{enumerate,enumitem}

\usepackage{csvsimple,booktabs}
\usepackage{float,subfig}
\usepackage{graphicx,epsfig,xcolor}
\usepackage{multirow}

\usepackage[colorlinks=true,breaklinks=true,bookmarks=true,urlcolor=MidnightBlue,citecolor=MidnightBlue,linkcolor=MidnightBlue,bookmarksopen=false,draft=false]{hyperref}
\usepackage{cite}

\usepackage{fullpage}

\usepackage{ebgaramond}
\usepackage[OT1]{fontenc}
\usepackage[utf8]{inputenc}

\usepackage{thmtools}
\declaretheorem{theorem}
\declaretheorem{proposition}

\declaretheorem[sibling=theorem]{lemma}
\declaretheorem{claim}

\input{header}

\newcommand{\MC}{\text{MC}}
\newcommand{\FA}{\text{FA}}
\newcommand{\D}{\text{D}}
\newcommand{\CD}{\text{CD}}

\newcommand{\TV}{\text{TV}}
\newcommand{\MO}{\text{MO}}
\newcommand{\pois}{\text{Poisson}}

\renewcommand{\hat}{\widehat}

\newcommand{\bose}[1]{{\textcolor{red}{Bose says: {#1}}}}

\renewcommand{\tilde}{\widetilde}
\renewcommand{\top}{{\mathpalette\raiseT\intercal}}

\title{
Controlling a Markov Decision Process with an \\ Abrupt Change in the Transition Kernel}

\author{Nathan Dahlin \qquad Subhonmesh Bose \qquad Venugopal V. Veeravalli
\thanks{All authors are with the Department of Electrical and Computer Engineering at the University of Illinois  Urbana-Champaign, Urbana, IL 61801. Emails: \{dahlin,boses,vvv\}@illinois.edu. This work was partly supported by a grant from C3.ai Digital Transformation Institute. }}

\begin{document}

\maketitle

\begin{abstract}
We consider the control of a Markov decision process (MDP) that undergoes an abrupt change in its transition kernel (mode). We formulate the problem of minimizing regret under control switching based on mode change detection, compared to a mode-observing controller, as an optimal stopping problem. Using a sequence of approximations, we reduce it to a quickest change detection (QCD) problem with Markovian data, for which we characterize a state-dependent threshold-type optimal change detection policy. Numerical experiments illustrate various properties of our control-switching policy. 
\end{abstract}

\section{Introduction}
\label{sec:introduction}
\input{introduction.tex}

\section{Formulating the Quickest Transition Kernel Change Detection Problem}
\label{sec:formulation}
\input{problem_formulation.tex}

\section{Defining the Regret}
\label{sec:formalizing_regret}
\input{true_additive_cost.tex}


\section{Approximating the Total Regret $\Rcal(\tau)$}
\label{sec:fast_mixing}
\input{mixing_argument.tex}

\label{sec:add_cost}
\input{simplified_additive_cost.tex}

\section{Optimal QCD Policy for Approximate Regret Minimization}
\label{sec:opt_qcd_approx}

\input{opt_qcd_cost.tex}
\section{Numerical Experiments}\label{sec:numerical}
\input{simulations.tex}


\section{Conclusions}
\label{sec:conc}
In this work, we examined the question of detecting a change in the transition kernel of an environment modeled as a finite state/action MDP in discrete time. We used an array of approximations to formulate a quickest change detection problem to minimize the approximate regret against a mode-observing controller. In our main result, we demonstrate that a threshold-type controller switching is optimal, where the thresholds are state-dependent.

There are a number of interesting directions for future research. First, we aim to consider the setting that allows switching among multiple modes, possibly multiple times. Second, we want to extend our results to continuous state/action MDPs via kernel methods, and consider minimax formulations that relax the geometric prior on the change point. Third, we want to analyze change detection-based reinforcement learning in non-stationary environments that do not assume knowledge of system model in various modes.

\bibliographystyle{alpha}
\bibliography{Bibliography}

\section{Appendices}
\input{appendix_regret.tex}
\input{appendix_fast_mixing_regret.tex}
\input{appendix_optimal_thresh.tex}

\input{appendix_uniqueness.tex}

\end{document}

%% file: header.tex
\newcommand{\beq}{\begin{equation}}
\newcommand{\eeq}{\end{equation}}
\newcommand{\beqa}{\begin{eqnarray}}
\newcommand{\eeqa}{\end{eqnarray}}
\newcommand{\beqan}{\begin{eqnarray*}}
\newcommand{\eeqan}{\end{eqnarray*}}

\newcommand{\E}{\mathds{E} }

\newcommand{\prob}{\mathbb{P}}

\newcommand\raiseT[2]{\raisebox{0.25ex}{$#1#2$}
}

\newcommand{\Rset}{\mathbb{R}}

\newcommand{\Uset}{\mathbb{U}}

\newcommand{\Xset}{\mathbb{X}}

\newcommand{\Zset}{\mathbb{Z}}

\newcommand{\Bcal}{{\cal B}}
\newcommand{\Ccal}{{\cal C}}

\newcommand{\Mcal}{{\cal M}}

\newcommand{\Pcal}{{\cal P}}

\newcommand{\Rcal}{{\cal R}}
\newcommand{\Scal}{{\cal S}}

\newcommand{\Zcal}{{\cal Z}}

\renewcommand{\v}[1]{{\bm{#1}}}

\newcounter{l1}
\newcounter{l2}
\newcounter{l3}
\newcommand{\bdotlist}{\begin{list}{$\bullet$}{}}
\newcommand{\bboxlist}{\begin{list}{$\Box$}{}}
\newcommand{\bbboxlist}{\begin{list}{\raisebox{.005in}{{\tiny
$\blacksquare$ \ \ }}}{}}
\newcommand{\bdashlist}{\begin{list}{$-$}{} }
\newcommand{\blist}{\begin{list}{}{} }
\newcommand{\barablist}{\begin{list}{\arabic{l1}}{\usecounter{l1}}}
\newcommand{\balphlist}{\begin{list}{(\alph{l2})}{\usecounter{l2}}}
\newcommand{\bAlphlist}{\begin{list}{\Alph{l2}.}{\usecounter{l2}}}
\newcommand{\bdiamlist}{\begin{list}{$\diamond$}{}}
\newcommand{\bromalist}{\begin{list}{(\roman{l3})}{\usecounter{l3}}}

%% file: introduction.tex

    
    
    
    
    
  

Control engineering has perfected the art of controller synthesis to optimize system costs, given a static environment, often modeled as an MDP. Practical engineering contexts are often non-stationary, i.e., the underlying dynamics of the environment encoded in the transition kernel of the MDP, changes at some point in time. A controller that is agnostic to that change can perform poorly against the changed environment. If the change point is known a priori, or is revealed when it happens, optimizing the control design becomes simple -- design controllers for the two ``modes'' of the system and then switch between the controllers, when the mode change happens. We tackle the question of mode change detection, followed by switching between mode-specific controllers.  Specifically, we consider the case where the underlying transition kernel of an MDP changes at a random point, and the change must be detected by observing the state dynamics. 

The mode of the system for our problem functions essentially as a hidden state. As a result, controlling a non-stationary MDP can be formulated as a partially-observed MDP (POMDP). Even when the state space is small and finite, solving a POMDP can be computationally challenging. See \cite{monahan1982state} for details. 
The change-detection based control switching paradigm provides a computationally cheaper alternative. 

Detection of a change in the statistics of a stochastic process has a long history in statistics and optimal control literature, with roots in \cite{shewhart1925application}.
Quickest change detection (QCD) theory seeks a causal control policy that an observer of the stochastic process can use to detect a change.
In designing such a detector, one must balance between two competing erroneous declarations. A hasty detector might declare the change too early, leading to a false alarm. With controller switching against non-stationary MDPs, an early change will incur extra cost of using a possibly sub-optimal controller over the period when the controller has changed, but the transition kernel has not. A lethargic detector who declares the change too late, pays the delay penalty of continuing to use the wrong controller, even after the transition kernel has changed. Thus, an optimal change detection mechanism balances between the possibilities of false alarm and delay, optimizing the costs incurred in these situations. In this paper, we formalize the controllers under the changing environment in Section \ref{sec:formulation} and define the regret of change detection-based controller switching in Section \ref{sec:formalizing_regret}.

When the pre- and post-change data is independent and identically distributed (i.i.d.), the optimal Bayesian change detector with geometric prior distribution assumed over the change point, was derived in \cite{shiryaev2007optimal}. Alternate formulations without such prior distributions have been derived in \cite{lorden1971procedures,ritov1990decision,pollak1985optimal}. See \cite{polunchenko2012state,veeravalli2014quickest} for detailed surveys on this history. In this paper, we assume a geometric prior on the change point. In Section \ref{sec:fast_mixing}, we simplify the regret expression for controller switching via approximations that rely on fast mixing of the Markov chains under the mode-specific controllers. Then in Section \ref{sec:opt_qcd_approx}, we formulate the question of optimal controller switching for minimizing the approximate regret as a QCD problem with Markovian data, and characterize an optimal change detector. Our main result (Theorem \ref{thm:optimal_thresh}) proves that the optimal switching policy is threshold-type, where the thresholds depend on the observed Markov state of the system. While our goal in this paper is (approximate) regret minimization with switching between mode-specific controllers for a single change in the transition kernel of an MDP, Theorem \ref{thm:optimal_thresh} applies more broadly to general QCD problems with Markovian observations. We remark that while our results are presented with change in transition kernels, the theory can be extended to consider changes in cost models across the change point.

In Section \ref{sec:numerical}, we present numerical experiments to  demonstrate the performance of the controller switching policy that we design on example non-stationary MDPs. Our results demonstrate that the change-detection based controller switching performs very similarly to that using mode-observation under a variety of parameter choices.


Perhaps the closest in spirit to our work is  \cite{banerjee2017quickest}, where the authors provide several interesting insights into the use of QCD to controller switching for non-stationary MDPs. Compared to \cite{banerjee2017quickest}, however, our primary goal is theoretical analysis to establish the  optimality of the  state-dependent threshold-type Bayesian change detector for approximate regret minimization, and to provide conditions under which such approximations are expected to be accurate. We conclude the paper in Section \ref{sec:conc}. Proofs of results are included in the appendix.

%% file: problem_formulation.tex
Consider a controlled, non-stationary Markovian dynamical system evolving over a finite state space $\Xset$ in discrete time ($t\in\Zset_+$) under the action of a controller choosing actions from a finite action space $\Uset$. The system is non-stationary in the sense that the state transition kernel remains fixed up until a random time $\Gamma$, when the system changes its mode. When the mode change happens, the state evolution, i.e., the transition kernel, changes. Denote the pre-change and post-change modes by $\theta=1$ and $\theta=2$, respectively. The corresponding transition kernels are given by $P_1(\cdot | x, u)$ and $P_2(\cdot | x, u)$, respectively, for $x\in\Xset$, $u \in\Uset$. The, non-stationary state transition kernel is then given by
\begin{equation}
    P_t(x_{t+1}|x_t,u_t) = \begin{cases}P_1(x_{t+1}|x_t,u_t),&\text{if }t<\Gamma,\\
    P_2(x_{t+1}|x_t,u_t),&\text{otherwise}.
    \end{cases}
\end{equation}
In general, the setting described thus far could additionally include mode dependent stage-cost functions. For simplicity, assume that stage-costs are not mode dependent, and, given by the time invariant function $c\,:\,\Xset\times\Uset\to\Rset$. Using the notation $Y_0^t$ to denote the sequence $(Y_0, \ldots, Y_t)$ for an arbitrary variable $Y$, we define $h_t := (x_0^t, u_0^{t-1})$ as the state-action history available prior to taking an action at time $t$. 

The objective is to choose a policy $\pi$ to minimize the long term cost, i.e., for all $x\in \Xset$,
\begin{equation}\label{eq:inf_horizon_cost}
    \begin{split}
        J^{\star}(x)
        &= \min_{\pi\in \Pi}\,J^{\infty}(x;\pi):= \min_{\pi\in\Pi}\,\E_{\substack{X_t \sim P_t(\cdot|x_{t-1}, u_{t-1})\\ U_t \sim \pi(\cdot | h_{t})}}\left[\sum_{t=0}^{\infty}\gamma^tc(X_t,U_t)|X_0=x\right].
    \end{split}
\end{equation}

As the change point $\Gamma$ is not observable, a Markovian policy may not, generally speaking, be optimal in \eqref{eq:inf_horizon_cost}. Considering $\theta_t\in\{1,2\}$, together with $x_t$ as part of an augmented state process $s_t:=(x_t,\theta_t)$, \eqref{eq:inf_horizon_cost} can be formulated as a partially observed Markov decision process (POMDP). Direct solutions to POMDPs are computationally intensive. We propose and analyze an alternative approach--employ change detection and switch between two mode-specific controllers, i.e.,
optimal stationary policies for each of the system modes. We denote the mode-specific policies as $\pi_i:\Xset\to\Uset$ for $i=1,2$, noting that in the fixed mode settings we may restrict to deterministic, Markovian policies without loss of optimality \cite{puterman2014markov}. In particular, $\pi_i\in\Pi$ for $i=1,2$, and 
\begin{equation}
{\small
    \begin{split}
        &J^{\infty}(x; \pi_i,P_i)=J^{\star}(x;P_i):= \min_{\pi \in \Pi} J^{\infty}(x; \pi,P_i):= \min_{\pi\in\Pi}\E_{\substack{X_t \sim P_i(\cdot|x_{t-1}, \pi(x_{t-1}))\\U_t\sim\pi(\cdot|h_t)}} \left[ \sum_{t=0}^\infty \gamma^t c(X_t, \pi(X_t))|X_0=x \right].
    \end{split}
}
\end{equation}

Limiting our search to policies that switch between $\pi_1$ and $\pi_2$ at a single time $t=\tau$, the more general policy design problem in \eqref{eq:inf_horizon_cost} reduces to an optimal stopping problem, where stopping corresponds to switching from $\pi_1$ to $\pi_2$. In this context, a natural performance baseline policy is one which observes the mode directly and switches controllers at time $t=\Gamma$. Intuitively, if an algorithm can accurately detect the change point, then a switching policy employing such an algorithm should well approximate the performance achieved when mode observations are available. Our goal is to utilize tools from the QCD literature to develop a change detection algorithm for the described non-stationary MDP setting, and characterize the regret of the resulting change-detection based controller with respect to the mode-observing baseline.

We remark that our regret analysis will not exactly compare a CD-based algorithm with the ``optimal'' control. To see why, note that $\pi_1$ is designed to minimize the long-run discounted cost, assuming that system remains in mode 1 for the entire infinite horizon. In other words, the design of $\pi_1$ is oblivious to the possibility of a mode change, and hence, is not strictly optimal in our non-stationary setting, even when $\theta_t=1$. The choice of operating with $\pi_2$, however, is optimal post-change, as we preclude the possibility of the mode switching back from $\theta=2$ to $\theta=1$. On the other hand, optimal policies arising from the POMDP formulation previously mentioned would take into account the potential mode change, and further, continuously adjust actions in accordance with a belief regarding the system mode. Tackling such complications is relegated to future endeavors.

%% file: true_additive_cost.tex
{Let $\v{X} = \{X_t\}_{t\geq 0}$ and $\v{X}'=\{X'_t\}_{t\geq0}$ denote the sequence of state observations generated by the change detection (CD) based and mode observing controllers, under the same random change point $\Gamma$. Note that the state trajectories of the CD and mode observing controllers will always agree prior to the change point, or the time at which the CD controller switches to $\pi_2$, whichever is earlier.\\
\indent We seek to identify a causal switching rule $\tau$ that defines an extended integer-valued random variable, adapted to the filtration $ \{\sigma(X_1,\dots,X_t)\}$, generated by the stochastic process. We will often use $\tau$ to denote the random stopping time chosen by the underlying decision rule. Note that we need not adapt $\tau$ to the sequence $\v{U} = \{U_t\}_{\{t\geq0\}}$, as the observed states are mapped deterministically to actions by both $\pi_1$ and $\pi_2$. 

Let $\Theta_t\in\{1,2\}$ denote the mode of the system at time $t$. Further, let $D_t$ denote the switching decision at time $t$, i.e., $D_t=1$ for $t<\tau$, and $D_t=2$ otherwise. Then, the regret minimization problem can be written as 
\begin{equation}\label{eq:true_add_cost}
    \min_{\tau\in\Scal}\,\,\E\left[\sum_{t=0}^{\infty}\gamma^tr_t(X_t,D_t,\Theta_t)\right]=:\Rcal(\tau),
\end{equation}
where $\Scal$ is the set of all causal switching policies satisfying $P(\tau<\infty)=1$, and 
\begin{equation}\label{eq:stage_regret}
    \begin{split}
        &r_t(X_t,D_t,\theta_t) = \begin{cases}
            c(X_t,\pi_1(X_t))-c(X'_t,\pi_2(X'_t)),& \text{if }D_t=1,\Theta_t=2,\\
            c(X_t,\pi_2(X_t))-c(X'_t,\pi_1(X'_t)),& \text{if }D_t=2,\Theta_t=1,\\
            c(X_t,\pi_2(X_t))-c(X'_t,\pi_2(X'_t)),& \text{if }D_t=2,\Theta_t=2.\\
        \end{cases}
    \end{split}
\end{equation}
\indent In order to further analyze the expected regret in \eqref{eq:true_add_cost}, we need to introduce additional notation. 
$J^m_{i|j}(Q)$ denotes the expected $m$ step horizon cost, when the initial state is distributed according to $Q$ and policy $\pi_i$ is used under system mode $j$, for $i,j=1,2$. If the initial state $X$ is given, we write $J^m_{i|j}(\delta_X)$, where $\delta_X$ is the measure concentrated at $X$. 
Let $P^{m}_{i|j}(\delta_X)$ be the $m$ step probability distribution, given that the system started in state $X$, and policy $\pi_i$ was used in system mode $j$. 
Additionally, we denote the stationary state distribution arising from the use of policy $\pi_i$ in mode $j$, for $i,j=1,2$ as $\Delta_{i|j}$, and the associated Markov chain as $\Mcal_{i|j}$. 
With this notation in hand, we can give expected regret terms corresponding to the false alarm and delay cases, conditioned on the history of state observations until time $t$.

For switching policies, $D_t=2$ implies that $D_{t+k}=2$ for all $k>0$. The decision to switch ends at the first time when $D_t=1$. Thus, we can fold the remaining infinite horizon cost-to-go into the cost of switching, and the optimal switching policy for \eqref{eq:true_add_cost} can be obtained from dynamic programming. Specifically, consider the Bellman equation 
\begin{equation}\label{eq:exact_Bellman}
\begin{split}
    V_t(X^t_0) &= \min_{d_t} \big\{\left(r^{\FA}_t + r^{\D}_t\right)\mathds{1}_{d_t=2}
    \\
    & \qquad + \big(\E\big[\mathds{1}_{\Gamma<t}(c(X_t,\pi_1(X_t))-c(X'_t,\pi_2(X'_t)))+ \gamma V_{t+1}(X^{t+1}_0)|X^t_0,d_t\big]\big)\mathds{1}_{d_t=1}\big\},
\end{split}
\end{equation}
where $r^{\FA}_t$ and $r^{\D}_t$ give the expected regret (cost) to go due to false alarm and delay when stopping at time $t$, respectively, conditioned on $X^t_0$. The following proposition gives expressions for $r^{\FA}_t$ and $r^{\D}_t$. See Appendix \ref{app:expected_regret} for a detailed derivation. Since $U_t = \pi_1(X_t)$ for $t<\tau$, we do not explicitly include conditioning on $\{U^t_0\}$ in \eqref{eq:exact_Bellman}. We likewise suppress inclusion of $\{U^t_0\}$ in the subsequent development in expressions where the use of $\pi_1$ or $\pi_2$ is clear from context. 
\begin{proposition}
\label{prop:expected_regret_to_go}
The following assertions hold:
\begin{align}
    &r^{\FA}_t = \E\bigg[\mathds{1}_{t<\Gamma}\bigg(\left(J^{\Gamma-t}_{2|1}(\delta_{X_t}) - J^{\Gamma-t}_{1|1}(\delta_{X_t})\right) \notag
    \\
    &\qquad \qquad \qquad
    + \gamma^{\Gamma-t}\left(J^{\infty}_{2|2}\left(P^{\Gamma-t}_{2|1}(\delta_{X_t})\right) - J^{\infty}_{2|2}\left(P^{\Gamma-t}_{1|1}(\delta_{X_t})\right)\right)\bigg)|X^t_0\bigg],
    \label{eq:prop_r_FA}
    \\
    &r^{\D}_t =\E\Bigg[\mathds{1}_{t\geq\Gamma}\Bigg(\left(J^{\infty}_{2|2}(\delta_{X_t}) - J^{\infty}_{2|2}(P^{t-\Gamma}_{2|2}(\delta_{X_{\Gamma}}))\right)\Bigg)|X^t_0\Bigg].\label{eq:prop_r_D}
\end{align}
\end{proposition}
Expressions \eqref{eq:prop_r_FA} and \eqref{eq:prop_r_D} provide further insight into the sources of expected regret when stopping at time $t$. In the false alarm scenario, the CD based controller will erroneously use $\pi_2$ in mode 1 for $\Gamma-t$ time steps. The regret incurred during this period consists of a transient component as $\Mcal_{2|1}$ and $\Mcal_{1|1}$ mix to $\Delta_{2|1}$ and $\Delta_{1|1}$, as well as a steady state component due to the difference between $\Delta_{2|1}$ and $\Delta_{1|1}$. Similar transient and steady state components will generally arise when either the mode changes, or the CD based controller declares a change.

Continuing in \eqref{eq:prop_r_FA}, there is a discrepancy between the distributions of the states $X_{\Gamma}$ and $X'_{\Gamma}$ due to the difference in control policies used by the CD and mode observing controllers for $\Gamma-t$ steps.
Similarly, when switching in the delayed detection scenario, there is a discrepancy in the distribution of states $X_t$ and $X'_t$, and thus the expected long term costs over the remaining horizon in mode 2. \\
\indent In terms of selecting a stopping rule, the form of \eqref{eq:prop_r_FA} and \eqref{eq:prop_r_D} introduce substantial difficulties. 
Application of dynamic programming techniques to \eqref{eq:exact_Bellman} requires explicit evaluation of \eqref{eq:prop_r_FA} and \eqref{eq:prop_r_D}, needing computation of finite and infinite horizon expected discounted costs, which often do not yield closed form expressions, requiring numerical estimation. As the change point $\Gamma$ is not observed, these discounted costs must be computed over ranges of potential values for $\Gamma$. Further, as \eqref{eq:prop_r_D} includes terms involving $X_{\Gamma}$, and the continuation option in \eqref{eq:exact_Bellman} involves an expectation over $X'_t$ assuming the change point occurred at some $\Gamma<t$, the entire history of state observations must be retained.
In light of these difficulties, we now develop approximations to $r^{\FA}_t$ and $r^{\D}_t$, which in turn yield a more tractable approximation of the expected regret in \eqref{eq:true_add_cost}.}

%% file: mixing_argument.tex
A primary component of the difficulty in direct use of \eqref{eq:prop_r_FA} and \eqref{eq:prop_r_D} lies in the state dependency of the discounted costs involved. This state dependency follows from the mixing properties of Markov chains associated with combinations of system modes and mode specific policies. Our problem simplifies significantly under the assumption that these Markov chains are fast mixing, i.e., the impact of an initial state or state distribution on expected future states is minimal. 
We have the following result on convergence, i.e., mixing of a Markov chain to its stationary distribution starting from an arbitrary distribution on $\Xset$. Let $\Pcal_{\Xset}$ denote the set of all probability measures on $\Xset$. 


\begin{theorem}[\cite{levin2017markov} Theorem 4.9]\label{thm:levin_mixing}
Let $\Mcal$ be an irreducible and aperiodic Markov chain on $\Xset$, with $t$-step probabilities $P^t$ and stationary distribution $\Delta$. Then there exists  $\beta\in(0,1)$ and $B>0$ such that 
\begin{equation}
    \sup_{\mu\in\Pcal_{\Xset}}\|\mu P^t-\Delta\|_{\TV}\leq B\beta^t.
\end{equation}
\end{theorem}

When there is a change in policy, mode or both, the overall system switches from one Markov chain into another, with initial state distribution affected by the preceding trajectory. Theorem \ref{thm:levin_mixing} allows us to bound the difference in cost to go between an arbitrary initial state distribution, and the stationary distribution of the Markov chain entered.
Let 
\begin{equation}
    c_i: = [c(x_{(1)},\pi_i(x_{(1)})),\dots,c(x_{(|\Xset|)},\pi_i(x_{(|\Xset|)}))]^{\top}\in \Rset^{|\Xset|}
    \label{eq:ci.def}
\end{equation}
denote the cost vector when using policy $i=1,2$, where $\{x_{(k)}\}_{k=1}^{|\Xset|}$ is an ordering of the states compatible with $\Delta_{i|j}$ for $i,j=1,2$.
Let ${E}^k(\mu,\Delta_{i|j})$ denote the difference between the $k$-step cost-to-go functions, given that the system enters $\MC(\Delta_{i|j})$ with initial state distributed according to $\mu$, i.e.,
\begin{equation}\label{eq:mixing_error}
\begin{split}
    {E}^k(\mu,\Delta_{i|j}) &:=|J^k_{i|j}(\mu) - J^{k}_{i|j}(\Delta_{i|j})|\\
    &= \left|J^k_{i|j}(\mu) - \frac{1-\gamma^{k}}{1-\gamma}c_i^{\top}\Delta_{i|j}\right|,
\end{split}
\end{equation}
where $c_i$ is defined in \eqref{eq:ci.def}. Assuming fast mixing, we neglect error terms of the form in \eqref{eq:mixing_error} and use the  approximation
\begin{equation}
    J^k_{i|j}(\mu) \approx J^k_{i|j}(\Delta_{i|j})= \frac{1-\gamma^k}{1-\gamma}c_i^{\top}\Delta_{i|j},
\end{equation}
 for $i,j=1,2$ and $k\leq\infty$.
An upper bound on the resulting error is then given in the following result that is proven in Appendix \ref{app:fast_mixing_regret}. 
\begin{lemma}\label{lem:mixing_error}
Let $\beta_{i|j}\in(0,1)$ and $B_{i|j}>0$ be the mixing coefficients of $\Mcal_{i|j}$ per Theorem \ref{thm:levin_mixing}. If $\mu\in\Pcal_{\Xset}$, then 
\begin{equation}
    {E}^k(\mu,\Delta_{i|j})\leq 2\|c_i\|_{\infty}B_{i|j}\frac{1-(\gamma\beta_{i|j})^{k}}{1-\gamma\beta_{i|j}}.
\end{equation}
 \end{lemma}

%% file: simplified_additive_cost.tex

Under the assumption of fast mixing, if we switch policies at time $t$, then the regret-to-go under false alarm and delay situations can be approximated using Proposition \ref{thm:levin_mixing} as follows.
\begin{equation}
    \begin{split}
        r^{\FA}_t &\approx\E\Bigg[\mathds{1}_{t<\Gamma}\bigg(\frac{1-\gamma^{\Gamma-t}}{1-\gamma}(c_2^{\top}\Delta_{2|1}-c_1^{\top}\Delta_{1|1})+ \frac{\gamma^{\Gamma-t}}{1-\gamma}c^{\top}_2(\Delta_{2|2}-\Delta_{2|2})\bigg)|X^t_0\Bigg]
        \\
        &= (c_2^{\top}\Delta_{2|1}-c_1^{\top}\Delta_{1|1})\E\left[\mathds{1}_{t<\Gamma}\frac{1-\gamma^{\Gamma-t}}{1-\gamma}|X^t_0\right],
    \end{split}
\end{equation}
\begin{equation}
    \begin{split}
        &r^{\D}_t\approx \E\Bigg[\mathds{1}_{t\geq\Gamma}\Bigg(\frac{\gamma^{t-\Gamma}}{1-\gamma}c_2^{\top}(\Delta_{2|2}-\Delta_{2|2})\Bigg)|X^t_0\Bigg]=0.
    \end{split}
\end{equation}
We are interested in analyzing the regret for discounted cost settings, where the discount factor is close to unity, as we seek to achieve good performance in both the pre- and post-change regimes. Smaller discount factors effectively penalize false alarms more heavily than delayed detection as the penalty on future costs are discounted more heavily. Thus, with $\gamma = 1-\epsilon$ with $\epsilon\ll 1$, we have
\begin{equation}
    \lim_{\epsilon\to 0}\frac{1-\gamma^{m}}{1-\gamma} = \lim_{\epsilon\to 0}\frac{1-(1-\epsilon)^{m}}{1-(1-\epsilon)} =  m
\end{equation}
for any $m > 0$, using which, we simplify $r^{\FA}_t$ as
\begin{equation}
    \begin{split}
        r^{\FA}_t&\approx (c^{\top}_2\Delta_{2|1}-c^{\top}_1\Delta_{1|1})\E\left[\mathds{1}_{t<\Gamma}(\Gamma-t)|X^t_0\right]= (c^{\top}_2\Delta_{2|1}-c^{\top}_1\Delta_{1|1})\E\left[(\Gamma-t)_+|X^t_0\right],
    \end{split}
\end{equation}
implying
\begin{equation}
    \begin{split}
        &\E[r^{\FA}_t] \approx (c^{\top}_2\Delta_{2|1}-c^{\top}_1\Delta_{1|1})\E[(\Gamma-t)_+].
    \end{split}
\end{equation}
Again with fast mixing, the single-step component of the continuation cost in \eqref{eq:exact_Bellman} can be approximated as
\begin{equation}
\begin{split}
    &\E\left[\mathds{1}_{\Gamma<t}\left(c(X_t,\pi_1(X_t))-c(X'_t,\pi_2(X'_t))\right)|X^t_0\right]\approx \E\left[\mathds{1}_{\Gamma<t}(c(X_t,\pi_1(X_t))-c^{\top}_2\Delta_{2|2})|X^t_0\right],
\end{split}
\end{equation}
since, given the event $\{\Gamma<t\}$, $X'_t$ will be distributed according to approximately $\Delta_{2|2}$. Approximating $\gamma\approx 1$, we may now write the following approximation of the Bellman equation in \eqref{eq:exact_Bellman},
\begin{equation}\label{eq:approx_Bellman}
{
\begin{split}
    \hat{V}(X^t_0) 
    &= \min_{d_t} \bigg\{(c^{\top}_2\Delta_{2|1}-c^{\top}_1\Delta_{1|1})\E[(\Gamma-t)_+|X^t_0] \mathds{1}_{d_t=2} 
    \\
    & \qquad \qquad + \E\bigg[\mathds{1}_{\Gamma<t}(c(X_t,\pi_1(X_t))-c^{\top}_2\Delta_{2|2}+ \hat{V}(X^{t+1}_0))|X^t_0,d_t\bigg]\mathds{1}_{d_t=1}\bigg\}.
\end{split}
}
\end{equation}
The optimal value functions $\hat{V}$ satisfy the Bellman equation of the following optimal control problem that optimizes over all causal control-switching policies $\Scal$
\begin{equation}\label{eq:interim_opt}
    \begin{split}
       &\min_{\tau\in\Scal}\,\,\E\Bigg[\sum_{t=\Gamma}^{\tau-1}c(X_t,\pi_1(X_t))-c^{\top}_2\Delta_{2|2}(\tau-\Gamma)_++(c^{\top}_2\Delta_{2|1}-c^{\top}_1\Delta_{1|1})(\Gamma-\tau)_+\Bigg]. 
    \end{split}
\end{equation}
Reapplying the fast mixing approximation, we have 
\begin{equation}\label{eq:last_delay_app}
    \E\left[\sum_{t=\Gamma}^{\tau-1}c(X_t,\pi_1(X_t))\right]\approx c^{\top}_1\Delta_{1|2}(\tau-\Gamma)_+
\end{equation}
for all $\tau \in \Scal$.
Combining \eqref{eq:interim_opt} and \eqref{eq:last_delay_app}, we arrive at the approximate regret minimization problem, given by
\begin{equation}
\begin{split}
    &\min_{\tau\in\Scal}\,\,(c^{\top}_2\Delta_{2|1}-c^{\top}_1\Delta_{1|1})\E[(\Gamma-\tau)_+]+ (c^{\top}_1\Delta_{1|2}-c^{\top}_2\Delta_{2|2})\E[(\tau-\Gamma)_+]=:\tilde{\Rcal}(\tau). 
\end{split}
\label{eq:approx_regret_min}
\end{equation}

In view of Lemma \ref{lem:mixing_error}, $\tilde{\Rcal}(\tau)$ will be a good approximation to $\Rcal(\tau)$, when the constants $\beta_{i|j}$ for $i,j=1,2$ are small. Roughly speaking, $\beta_{i|j}$ is small when the spectral gap of $\Mcal_{i|j}$ is large, i.e. $\beta_{i|j}$ depends largely on the second leading eigenvalue of $\Mcal_{i|j}$ (see \cite{levin2017markov} Chapter 12). Additionally, a switching policy optimizing \eqref{eq:approx_regret_min} will prioritize minimization of false alarm and delay based regret based upon the difference in average stage costs in the involved Markov chains. For instance, when $c^{\top}_2\Delta_{2|1}-c^{\top}_1\Delta_{1|1}\ll c^{\top}_1\Delta_{1|2}-c^{\top}_2\Delta_{2|2}$, the controller may exhibit longer average false alarm periods. However, if the expected unit step loss in erroneously using $\pi_2$ in mode 1 is not too large, then assuming fast mixing, the aggregate regret due to such errors should likewise be small. 
In what follows, we seek to identify stopping rules that exactly solve \eqref{eq:approx_regret_min}.
Assume a geometric prior on the change point $\Gamma$ with $\prob\{\Gamma=t\}=\rho(1-\rho)^{t-1}$ for $t\geq 1$. Then, the memoryless property of geometric distributions yields $ \E[(\Gamma-\tau)_+] 
    = P\{\Gamma>\tau\}/\rho$.
Assuming $c^{\top}_1\Delta_{1|2}-c^{\top}_2\Delta_{2|2}>0$, we further get
\begin{equation}\label{eq:rescaled_Rcal_approx}
    \begin{split}
        \tilde{\Rcal}(\tau) \propto \E[(\tau-\Gamma)_+] + \lambda P\{\Gamma>\tau\},
    \end{split}
\end{equation}
where 
\begin{equation}\label{eq:rescaled_lambda}
    \lambda\rho = \frac{c^{\top}_2\Delta_{2|1}-c^{\top}_1\Delta_{1|1}}{c^{\top}_1\Delta_{1|2}-c^{\top}_2\Delta_{2|2}}.
\end{equation}

The above argument requires both the numerator and the denominator in \eqref{eq:rescaled_lambda} to be non-negative. Notice that when $X_0 \sim \mu_0$, we have
\begin{align}
\begin{aligned}
    J^{\infty}_{1|1}(\mu_0)
    &= \frac{1}{1-\gamma} c_1^\top \Delta_{1|1} + {E}^\infty(\mu_0, \Delta_{1|1}) \leq  \frac{1}{1-\gamma} c_2^\top \Delta_{2|1} + {E}^\infty(\mu_0, \Delta_{2|1}),
\end{aligned}
\end{align}
owing to the optimality of $\pi_1$ for the MDP with in mode 1.
In essence, this implies that 
\begin{align}
    c^{\top}_2\Delta_{2|1}-c^{\top}_1\Delta_{1|1} \geq (1-\gamma)\left[ {E}^\infty(\mu_0, \Delta_{1|1}) - {E}^\infty(\mu_0, \Delta_{2|1}) \right].
\end{align}
As long as $\Mcal_{2|1}$ and $\Mcal_{1|1}$ are fast-mixing, we expect the errors due to mixing to be small, and the left-hand side to be positive. A similar argument applies to \eqref{eq:rescaled_lambda}'s denominator.



%% file: opt_qcd_cost.tex
This section is devoted to the derivation of the optimal change detector to minimize the approximate regret $\tilde{\Rcal}$. Specifically, this regret under a causal control switching policy $\tau \in \Scal$ in \eqref{eq:approx_regret_min} can be written as
\begin{equation}\label{eq:Rcal_additive}
    \tilde{\Rcal}(\tau) \propto \E\left[\sum_{t=0}^{\tau}\tilde{g}_t(D_t,\Theta_t)\right],
\end{equation}
where $X_0 \sim \mu_0$ and
\begin{equation}\label{eq:rescaled_g}
    \tilde{g}_t(D_t,\Theta_t) = \begin{cases}
    1, &\text{if }D_t=1,\Theta_t=2,\\
    \lambda, &\text{if }D_t=2,\Theta_t=1, \\
    0, &\text{otherwise}.
    \end{cases}
\end{equation}
The optimal control problem in \eqref{eq:Rcal_additive} is amenable to infinite horizon dynamic programming. For $t\geq 1$, define
\begin{equation}
    p_t := \prob\{\Gamma\leq t-1|X_0^t\}=\prob\{\Gamma< t|X_0^t\}
\end{equation}
as the posterior probability at time $t$ that the change has taken place prior to time $t$, with $p_0=0$.
This probability is a sufficient statistic in the sense that $(X_t, p_t)$ defines a Markovian state for dynamic programming calculations for this optimal control problem, following  \cite{bertsekas2012dynamic}.

Let $\pi^{\CD}$ denote the control policy that uses change detection to switch between $\pi_1$ and $\pi_2$. Specifically, we have
\begin{align}
    \pi_{\CD}(X_t) 
    = \begin{cases}
    \pi_1(X_t), &\text{if } t<\tau,
    \\
    \pi_2(X_t), &\text{otherwise}.
    \end{cases}
\end{align}
Then, define
\begin{align}
    L(X_{t+1},X_t) := \frac{P_2(X_{t+1}|X_t,\pi_{\CD}(X_t))}{P_1(X_{t+1}|X_t,\pi_{\CD}(X_t))}.
    \label{eq:L.def}
\end{align}
Recall that $P_i$ encodes the transition kernel in mode $\theta=i$. Thus, $L$ becomes the likelihood ratio of observing $X_{t+1}$ in mode 2 compared to that in mode 1. Further, define 
\begin{align}
        \bar{p}_t := p_t + \rho(1-p_t).
        \label{eq:p_tilde_def}
\end{align}
Then, the evolution of $p_t$ via Bayes' rule becomes
\begin{align}
    \label{eq:p_recursion}
    \hspace{-0.1in}p_{t+1} = \Phi(X_{t+1},X_t,p_t) := \frac{\bar{p}_t L(X_{t+1},X_t)}{\bar{p}_t L(X_{t+1},X_t) + (1-\bar{p}_t)}.
\end{align}
Per the dynamic programming principle, the optimal cost of the minimization of the right hand side of \eqref{eq:Rcal_additive} is given by $\E_{X_0\sim \mu_0} [\tilde{V}(0, X_0)]$, where the cost-to-go function $\tilde{V}$ satisfies the Bellman equation,
\begin{equation}
    \begin{split}
        &\tilde{V}(p_t,X_t) = 
        \min_{d_t}\, \left\{\lambda(1-p_t)\mathds{1}_{d_t=2} + \Big(p_t+ \E[\tilde{V}(p_{t+1},X_{t+1})|p_t,X_t,d_t]\Big)\mathds{1}_{d_t=1}\right\}.
    \end{split}
    \label{eq:qcd_Bellman}
\end{equation}
The expectation in the right hand side of the above equation is calculated as follows. First, fix $\pi_{\CD}(X_t) = \pi_1(X_t)$. Then, the distribution of $X_{t+1}$ is either that dictated by $P_1$ or $P_2$, depending on whether the change point has happened before $t$, while the action is taken according to $\pi_1(X_t)$. For each candidate $X_{t+1}$, evaluate $L$ using \eqref{eq:L.def} with $\pi_{\CD}(X_t) = \pi_1(X_t)$ and $p_{t+1}$ via \eqref{eq:p_recursion}.

The following result, sketched in Appendix \ref{app:optimal_thresh}, characterizes the nature of the optimal change detection policy. The result needs additional notation. Let $\Ccal$ denote the set of nonnegative functions of $p\in[0,1]$, dominated by $\lambda(1-p)$. 
Define the right hand side of \eqref{eq:qcd_Bellman} as $\Bcal\tilde{V}$ for $\Bcal:\Ccal\to\Ccal$.
\begin{theorem}\label{thm:optimal_thresh}
There exists a unique solution to the fixed-point equation $\Bcal\tilde{V} = \tilde{V}$ in $\Ccal$, which is given by $\tilde{V} = \lim_{k\to \infty} \Bcal^k \psi$, where $\psi(p,x) := \lambda(1-p)$, for $p\in[0,1]$ and $x \in \Xset$.
Moreover, an optimal control-switching policy to minimize $\tilde{\Rcal}(\tau)$ is threshold-type, given by $    \tau^{\star} = \min\{t:p_t\geq {p}^\circ(X_t)\}$,
where $p^{\circ}(X_t)\in[0,1]$ is the unique solution of 
\begin{align}
    \begin{aligned}
    &\lambda(1-{p}^\circ(X_t))  - {p}^\circ(X_t)=
    \E[\tilde{V}(\Phi(X_{t+1},X_t,p^{\circ}(X_t)), X_{t+1})|{p}^\circ(X_t),X_t,d_t=1].
    \end{aligned}
\end{align}

\end{theorem}

We prove the existence of value functions that satisfy the Bellman equation. In discounted-cost optimal control problems, such existence is typically argued via  the Contraction Mapping Theorem. Our QCD problem for approximate regret is undiscounted; a classic contraction mapping-based argument does not work. Our proof structure shows the result over a finite horizon first, and then extends it to the infinite-horizon setting using an argument similar to that used in \cite{veeravalli1993decentralized}.

Next, we contextualize this result within prior art on QCD. Per \cite{shiryaev2007optimal, veeravalli2014quickest}, Bayesian QCD with i.i.d. data under a geometric prior on the change point yields a single-threshold stopping rule. Unlike that setting, our data is Markovian. Theorem \ref{thm:optimal_thresh} reveals that a similar single-threshold policy is optimal with Markovian data too, but with the modification that the threshold is state-dependent. While we derive this stopping rule by viewing $\tilde{\Rcal}(\tau)$ as the approximate regret in switching control policies for MDPs under an environment with a changing transition kernel, the result applies more generally to Bayesian change detection with Markovian observations, and might be of independent interest to the QCD literature.

%% file: simulations.tex
We now illustrate properties of our controller switching through numerical examples. In our simulations, the  finite horizon and the geometric prior were such that pre- and post-change Markov chains mix sufficiently.

\subsection{On a Random MDP Environment}
We simulated the dynamics of a non-stationary MDP with $|\Xset| = 5$ and $|\Uset|=3$. Entries of $P_1$ were generated from $\text{Unif}[0,1]$, and then rescaled to yield a transition kernel. $P_2$ was obtained by permuting $P_1$ with respect to actions, e.g., $P_1(x'|x,a=0) = P_2(x'|x,a=1)$, and so on. Costs for each state/action pair were sampled from $\text{Unif}[0,1]$. We used a geometric prior for a range of $\rho$'s, listed in Table \ref{tab:MDP_stats}. With $\gamma=0.999$, we used value iteration to compute $\pi_1$ and $\pi_2$, by uniformly discretizing the posterior probability values to 1000 points in [0,1]. The change detection-based controller was computed via repeated application of the Bellman operator $\Bcal$ on $\psi$. We recorded the discounted costs over horizons of lengths $H = \lceil2/\rho\rceil$, averaged over 6000 runs, calculated with mode observation-based controller switching in $J_{\MO}$ and that based on change detection in $J_{\CD}$. The difference between the mean costs are all significant ($p$-values $< 0.02$). $J_{\CD}$ is within 0.7\% of $J_{\MO}$, implying that change detection-based switching performs quite well, compared to mode-observed switching.


With the right hand side of \eqref{eq:rescaled_lambda} held constant, $\lambda$--the relative weight of regret due to false alarm compared to that due to delay--varies inversely with $\rho$. Thus, as $\rho$ increases, $\lambda$ decreases, and the state dependent thresholds $p^\circ$ for various states decrease as Figure \ref{fig:thresh_PFA} reveals. In essence, the algorithm becomes more keen to switch from $\pi_1$ to $\pi_2$, incurring possibly higher regret due to false alarm that is penalized less. Figure \ref{fig:thresh_PFA} confirms that the probability of false alarm grows with $\rho$.

\begin{figure}[h]
\centering
\includegraphics[width=0.5\textwidth]{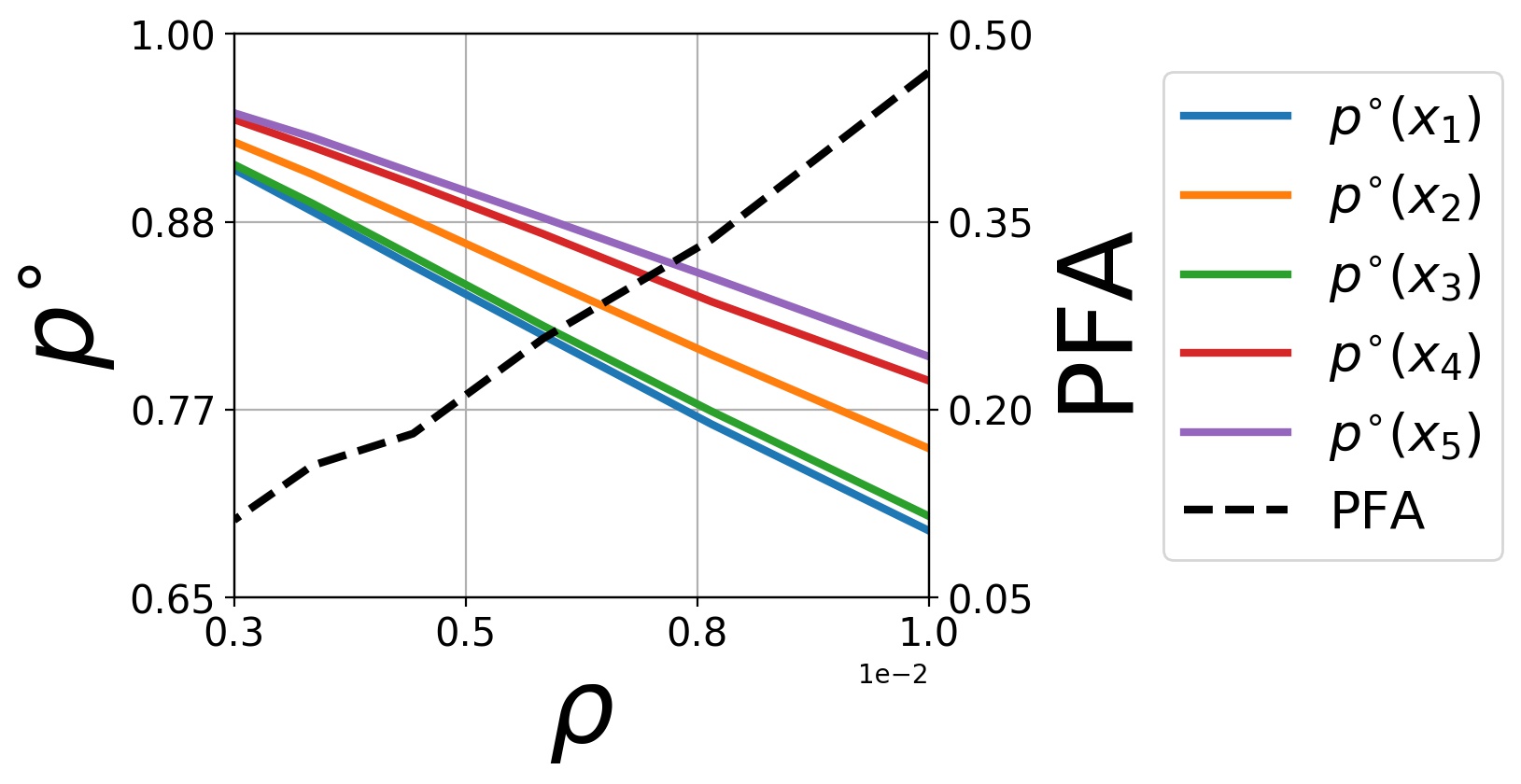}
\caption{State dependant switching thresholds and probability of false alarm (PFA) for the random MDP.}
\label{fig:thresh_PFA}
\end{figure}


\subsection{Inventory Control Problem}
Next, we present results on a non-stationary environment, studied in \cite{banerjee2017quickest}. The (random) state $X_t$ represents the inventory level at time $t$ with maximum inventory level $N$. Let $U_t$ denote the additional inventory to be ordered, yielding total inventory $I_t:=\min\{X_t + U_t, N\}$. Having selected $U_t$, the random integral demand $W_t$ is realized, giving residual inventory $X_{t+1}=\max\{0,X_t+U_t-W_t\}$. With $v$, $h$ and $d$ denoting the costs per unit ordered, inventory held, and demand lost,
%
\begin{equation}
    c(X_t,U_t) := v X_t + h[I_t-W_t]_+ + d[W_t-I_t]_+.
\end{equation}
The demands are distributed according to $W_t\sim\pois(\nu)$ for $t< \Gamma$, and $W_t\sim\text{Unif}[0,N]$ afterwards. Conditioned on the change point, the demands are independent. We assume a  geometric prior on the change point with $\rho=0.01$. We run tests with $N\in\{10,15\}$ and $d\in\{100,200,300\}$. The remaining parameters are set to $v=1$, $h=5$, $\gamma=0.999$, $\nu=2$. For this example, define 
\begin{equation}
    \hspace{-0.1in}c_{i|j} := [\E_j[c(0,\pi_i(0)],\dots,\E_j[c(N,\pi_i(N))]]^{\top}\in\Rset^{N+1},
\end{equation}
where $\E_j$ stands for expectation with respect to the demand distribution in mode $j\in\{1,2\}$. In this example, the cost vectors depend on the mode as well as the policy, giving 
\begin{equation}\label{eq:lambda_mode_dep}
    \lambda\rho = \frac{c^{\top}_{2|1}\Delta_{2|1}-c^{\top}_{1|1}\Delta_{1|1}}{c^{\top}_{1|2}\Delta_{1|2}-c^{\top}_{2|2}\Delta_{2|2}},
\end{equation}
differing from \eqref{eq:rescaled_lambda} with mode-dependent costs.
We compute the optimal controllers in each mode via value iteration, determine $\lambda$ via \eqref{eq:lambda_mode_dep} and then approximate thresholds in Theorem \ref{thm:optimal_thresh} by discretizing $p_t$ values to 100 uniformly spaced points in $[0,1]$, and iterating $\Bcal$. 

\begin{table}[ht]
    \centering
    \caption{Performance comparison for random MDP environment of the discounted aggregate costs with mode-observed and change detection-based controller switching, averaged over 6000 runs over horizon length $H=2/\rho$.}
{
\begin{tabular}{c c c}
        \toprule
         $\rho$&$J_{\MO}$&$J_{\CD}$  \\ \midrule
         0.0100 &67.68 &68.04 \\ 
 0.0078 &85.03 &85.33\\
 0.0060 &105.81 &106.46 \\
 0.0046 &130.87 &131.34 \\
 0.0036 &159.97 &160.59 \\
 0.0028 &192.82 &193.44 \\
         \bottomrule
    \end{tabular}
}
    \label{tab:MDP_stats}
\end{table}

We recorded the discounted costs over horizons of length 1000, averaged over 4000 episodes. Across all cost and inventory parameters tested, the difference between mean costs $J_{\MO}$ and $J_{\CD}$ is less than 3.4\% of $J_{\MO}$, and in cases where $d\in\{100,300\}$, the difference is less than 0.8\% of $J_{\CD}$. Again, the difference between mean costs are significant ($p$-values $<0.002$). Thus, the change detection based switching performs well with respect to the mode observation-based controller. Note that in this simulation we do make use of cost observations in the change detection procedure, although the cost distribution changes with the mode. While incorporation of stage cost observations may improve the change detection performance, we leave such an extension to future work.

\begin{table}[ht]
    \centering
    \caption{Mean performance in the inventory control problem for the mode observing and change detection-based switching controllers with $\gamma=0.999$ and 4000 episodes.}
{
\begin{tabular}{c c c c c}
        \toprule
         $N$ & $d$ & $\lambda$ & $J_{\MO}$& $J_{\CD}$ \\ \midrule
         10 & 100 &19.39&17791&17864\\ 
 10 & 200 &8.06 &17832&18302\\
 10 & 300 &7.10&18280&18434\\
 \midrule
          15 & 100 &15.49&25941&26131\\
 15 & 200 &6.97&25944&26837\\
 15 & 300 &5.33&26404&26626\\
         \bottomrule
    \end{tabular}
}
    \label{tab:inv_stats}
\end{table}

%% file: appendix_regret.tex
\subsection{Proof of Proposition \ref{prop:expected_regret_to_go}}
\label{app:expected_regret}
Conditioning the expected false alarm regret-to-go on $X_0^t$,
\begin{align}
\label{eq:cond_exp_FA_cost}
    {
    \small
    \begin{aligned}
        &\E\left[r^{\FA}_t|X_0^t\right]\\
         &=\E\left[\mathds{1}_{t<\Gamma}\Bigg( \sum_{n=t}^{\Gamma-1}\gamma^{n-t}\left[c(X_n,\pi_2(X_n)) - c(X'_n,\pi_1(X'_n))\right] \right. 
         \\
         & \qquad \qquad \qquad\left.
         + \sum_{n=\Gamma}^{\infty}\gamma^{n-t}\left[c(X_n,\pi_2(X_n))-c(X'_n,\pi_2(X'_n))\right]\Bigg)|X^t_0\right]
        \\
        &=\E\left[\mathds{1}_{t<\Gamma}\E\left[\Bigg( \sum_{n=t}^{\Gamma-1}\gamma^{n-t}\left[c(X_n,\pi_2(X_n)) - c(X'_n,\pi_1(X'_n))\right] \right. \right.
        \\
        & \qquad\qquad \qquad \left. \left. + \sum_{n=\Gamma}^{\infty}\gamma^{n-t}\left[c(X_n,\pi_2(X_n))-c(X'_n,\pi_2(X'_n))\right]\Bigg)|X^t_0, \Gamma\right] | X_0^t\right]
        \\
        &=\E\left[\mathds{1}_{t<\Gamma}\left(J^{\Gamma-t}_{2|1}(\delta_{X_t}) - J^{\Gamma-t}_{1|1}(\delta_{X_t}) + \gamma^{\Gamma-t}\left(J^{\infty}_{2|2}(P^{\Gamma-t}_{2|1}(\delta_{X_t})) - J^{\infty}_{2|2}(P^{\Gamma-t}_{1|1}(\delta_{X_t}))\right)\right)|X^t_0\right],
    \end{aligned}
    }
\end{align}
using the law of total expectation.
Proceeding similarly with the expected delay regret-to-go,
\begin{equation}\label{eq:cond_exp_D_cost}
    {
    \begin{split}
       \E\left[r^{\D}_t|X^t_0\right] &=\E\Bigg[\mathds{1}_{\Gamma< t} \sum_{n=t}^{\infty}\gamma^{n-t}\left[c(X_n,\pi_2(X_n))-c(X'_n,\pi_2(X'_n))\right]|X^t_0\Bigg]\\
        &= \E\Bigg[\mathds{1}_{\Gamma< t}\E\Bigg[\Bigg( \sum_{n=t}^{\infty}\gamma^{n-t}[c(X_n,\pi_2(X_n))-c(X'_n,\pi_2(X'_n))]\Bigg)|X^t_0,\Gamma\Bigg]|X^t_0\Bigg]\\
        &= \E\left[\mathds{1}_{\Gamma<t}\left(J^{\infty}_{2|2}(\delta_{X_t}) - J^{\infty}_{2|2}(P^{t-\Gamma}_{2|2}(\delta_{X_{\Gamma}}))\right)|X^t_0\right].
    \end{split}
    }
\end{equation}

%% file: appendix_fast_mixing_regret.tex
 \subsection{Proof of Lemma \ref{lem:mixing_error}}\label{app:fast_mixing_regret}
 For any distribution $\tilde{\mu}$ over the initial state $X_0$, in mode $i$ under policy $j$, the (discounted) $T$-step cost to go is given by
\begin{align}
\begin{aligned}
     J^{T}_{i|j}(\tilde{\mu}) &= \E_{X_0 \sim \tilde{\mu}}\left[\sum_{t=0}^{T-1}\gamma^tc(X_t,\pi_i(X_t))\right]
     = \sum_{t=0}^{T-1}\gamma^t\E_{X_0 \sim \tilde{\mu}}\left[c(X_t,\pi_i(X_t))\right]
    = \sum_{t=0}^{T-1}\gamma^tc^{\top}_iP^{t}_{i|j}(\tilde{\mu}).
\end{aligned}
\end{align}
Therefore, we have
\begin{align}
{\small
\begin{aligned}
    |J^{T}_{i|j}(\tilde{\mu}) - J^{T}_{i|j}(\Delta_{i|j})| 
    &\leq \|c_i\|_{\infty}\sum_{t=0}^{T-1}\gamma^t\|P^{t}_{i|j}(\tilde{\mu}) - \Delta_{i|j}\|_1\\
    &= 2\|c_i\|_{\infty}\sum_{t=0}^{T-1}\gamma^t\|P^{t}_{i|j}(\tilde{\mu}) - \Delta_{i|j}\|_{\TV}
    \\
    &\leq 2 B_{i|j} \|c_i\|_{\infty}\frac{1-(\gamma\beta_{i|j})^{T}}{1-\gamma\beta_{i|j}}.
\end{aligned}
}
\end{align}

%% file: appendix_optimal_thresh.tex
\subsection{Proof of Theorem \ref{thm:optimal_thresh}}\label{app:optimal_thresh}
Consider the minimization of $\tilde{R}(\tau)$ over a finite horizon of length $T$. Precisely, define the optimal cost of the following $T$-period optimal control problem,
\begin{align}
\begin{aligned}
            \hspace{-0.5in} \tilde{v}^T_\star(p,x) := \min_{\tau\in\Scal^T}\E\left[\sum_{t=0}^{\tau}\tilde{g}_t(D_t,\Theta_t)|p_0=p,X_0=x\right].
\end{aligned} 
\end{align}
where $\Scal^T := \{\tau\in\Scal:P\{\tau\leq T-1\}=1\}$. Here, $\tilde{v}^T_\star(p,x)$ describes the optimal cost, assuming a starting state $x$ with an initial belief $p$ that $\theta_0=2$ for all $t$. The sequence $\{p_t\}$ with $p_0=p$ evolves according to the recursion $\Phi$ defined in \eqref{eq:p_recursion}.
In this $T$-period problem, if a change has not been declared at any time $t<T$, then it is declared at time $T$. As $\Scal^T\subseteq \Scal^{T+k}$ for $k>0$, $\tilde{v}^\star_T$ is decreasing in $T$. The optimal costs satisfy
\begin{subequations}
\begin{align}
    \tilde{v}^0_\star(p,x) &= \lambda(1-p),
    \\
    \tilde{v}^{T+1}_\star(p,x) &= \min\{\lambda(1-p), p + \E[\tilde{v}^T_\star(\Phi(X^+,x,p),X^+)|p,x,\pi_1(x)]\}.
    \label{eq:val_func_forward}
\end{align}
\end{subequations}
for $0<t\leq T-1$.
The relation in \eqref{eq:val_func_forward} holds as the change may be declared at time $t=0$, or at a later time $t>0$. When selecting the latter option, due to the memoryless property of the geometric prior on the change point, the optimal cost following the transition due to state-action pair $(x,\pi_1(x)$) is precisely $\tilde{v}^T_\star(p^+,X^+)$. 

Define the optimal cost over the infinite horizon as
\begin{equation}
    \tilde{v}_\star(p,x) := \inf_{\tau\in\Scal}\E\left[\sum_{t=0}^{\tau}\tilde{g}_t(D_t,\Theta_t)|p_0=p,X_0=x\right]. 
\end{equation}
We now show that 
\begin{align}
    \tilde{v}_\star = \lim_{T\to\infty}\tilde{v}^T_\star.
    \label{eq:v.lim}
\end{align} 
The limit exists, owing to the decreasing and nonnegative nature of $\tilde{v}^T_\star$.
Since $\Scal^T \subset \Scal$, we have $\tilde{v}\leq\lim_{T\to\infty}\tilde{v}^T_\star$.  To show the reverse inequality, fix $(p,x)$ and  $\varepsilon > 0$. Then, there exists an almost surely finite stopping time $\tau_{\varepsilon}\in \Scal$ such that 
\begin{equation}
    \tilde{v}_\star(p,x) + \varepsilon\geq \E\left[\sum_{t=0}^{\tau_{\varepsilon}}\tilde{g}_t(D_t,\Theta_t)|p_0=p,X_0=x\right].
\end{equation}
We then have that 
\begin{equation}
\begin{split}
    \E\left[\sum_{t=0}^{\tau_{\varepsilon}}\tilde{g}_t(D_t,\Theta_t)|p_0=p,X_0=x\right] &\overset{(a)}{=}\E\left[\lim_{T\to\infty}\sum_{t=0}^{\min\{\tau_{\varepsilon},T\}}\tilde{g}_t(D_t,\Theta_t)|p_0=p,X_0=x\right]\\
    &\overset{(b)}{=}\lim_{T\to\infty}\E\left[\sum_{t=0}^{\min\{\tau_{\varepsilon},T\}}\tilde{g}_t(D_t,\Theta_t)|p_0=p,X_0=x\right]\\
    &\geq \lim_{T\to\infty}\tilde{v}^T_\star(p,x),
\end{split}
\end{equation}
where (a) uses the monotone convergence theorem and $P\{\tau_{\epsilon}<\infty\} =1$, while (b) uses the dominated convergence theorem. 
Since $\varepsilon > 0$ is arbitrary, $\lim_{T\to\infty}\tilde{v}^{T}_\star(p,x)\leq \tilde{v}_\star(p,x)$, which completes the proof of  $\tilde{v}_\star = \lim_{T\to\infty}\tilde{v}^T_\star$.

For the $T$-period problem, dynamic programming implies the existence of optimal cost-to-go functions $\tilde{V}_0^T, \ldots \tilde{V}_T^T$ that satisfies 
$\tilde{v}^T_\star(p,x) = \tilde{V}^T_0(p,x)$, where
\begin{subequations}
\begin{align}
    \tilde{V}^T_T(p,x) &= \lambda(1-p),
    \label{eq:DP_suff_stage_cost_T}\\
    \tilde{V}^T_t(p,x) &= \min\left\{\lambda(1-p),p + A^T_t(p,x)\right\}
         \label{eq:DP_suff_stage_cost}
\end{align}
\end{subequations}
for $t=0, \ldots, T-1$, where
\begin{align}
    A^T_t(p,x) := \E\left[\tilde{V}^T_{t+1}(\Phi(X^+, x, p),X^+)|p,x,\pi_1(x)\right].
    \label{eq:def.AT}
\end{align}
Here, $X^+$ denotes the random next state from $x$ under action $\pi_1(x)$, given belief $p$ on the change. With this notation, we prove the result through the following sequence of steps.
\begin{enumerate}[label=(\alph*)]
    \item We first show that $\tilde{V}^T_t(p,x)$ for $t=0, \ldots, T$ and $A^T_t(p,x)$ for $t=0, \ldots, T-1$ are nonnegative, concave functions of $p\in[0,1]$ for every $x\in \Xset$. 
    
    \item Next, we utilize the concavity of $\tilde{V}^T$ and $A^T$ to deduce the threshold structure of an optimal stopping rule for the $T$-period problem.
    
    \item Then, we establish that $
    \lim_{T \to \infty}\tilde{V}_t^T$ exists and the limit satisfies $\tilde{V} = \tilde{v}_\star$.
    
    \item If $\Bcal G = G$ for any $G \in \Ccal$, then we show that $G=\tilde{V}$.
    
    \item We conclude the proof by showing $\tilde{V} = \lim_{k\to \infty}\Bcal^k \psi$.
\end{enumerate}

\subsubsection*{$\bullet$ Step (a). Showing the nonnegative concave nature of $\tilde{V}^T(\cdot, x)$ and $A^T(\cdot, x)$}

We prove the claim using backward induction. Notice that $\tilde{V}^T_T(p,x)$ does not depend on $x$, and is a nonnegative, affine function of $p\in[0,1]$. Hence, the claim regarding $\tilde{V}$'s holds for $\tilde{V}^T_T(p,x)$. Then, we have
\begin{equation}
\label{eq:A_T_T_1}
    \begin{split}
        A^T_{T-1}(p,x) 
        =\lambda - \lambda\mathbb{E}\left[\Phi(X^+,x,p)|p,x,\pi_1(x)\right].
        \end{split}
\end{equation}
To simplify the above expression, notice that 
\begin{equation}\label{eq:phi_exp}
    \begin{split}
        &\mathbb{E}[\Phi(X^+,x,p)|p,x,\pi_1(x)] = \sum_{x^+\in\Xset}\Phi(x^+,x,p)P(x^+|p,x,\pi_1(x)),
    \end{split}
\end{equation}
where $P$ encodes the probability of $X^+$, given $p, x, \pi_1(x)$. Since $p$ equals the probability of the mode having switched already by time $t$, we infer
\begin{align}
\begin{aligned}
    P(x^+|p, x, \pi_1(x))
    &= (1-p)\left[(1-\rho) P_1(x^+|x, \pi_1(x)) + \rho P_2(x^+|x, \pi_1(x)) \right] 
    + p P_2(x^+|x, \pi_1(x))
    \\
    &= \left[1 - \left( p + \rho (1-p) \right) \right] P_1(x^+|x,\pi_1(x)) + \left( p + \rho (1-p) \right) P_1(x^+|x,\pi_1(x))
    \\
    &= (1-\bar{p}) P_1(x^+|x,\pi_1(x))
    + \bar{p} P_2(x^+|x,\pi_1(x)).
\end{aligned}
\end{align}
Then, \eqref{eq:L.def} and \eqref{eq:p_recursion} in the above relation yield
\begin{equation}
\hspace{-0.1in}
    \begin{split}
        \E[\Phi(X^+,x,p)|p,x,\pi_1(x)]
        &=\sum_{x^+\in\Xset}\frac{\bar{p} L(x^+,x)}{\bar{p} L(x^+,x) + (1-\bar{p})}P(x^+|p,x,\pi_1(x))
        \\
        &= \sum_{x^+\in\Xset}\bar{p} P_1(x^+|x,\pi_1(x))
        \\
        &=\bar{p}.
    \end{split}
    \label{eq:exp.phi}
\end{equation}
Further, using \eqref{eq:exp.phi} in \eqref{eq:A_T_T_1}, we get
\begin{equation}
    A^T_{T-1}(p,x)
    =\lambda(1-\rho)(1-p),
\end{equation}
which is a nonnegative and affine (concave) in $p$. This step completes the base case for the induction argument. 

Next, assume that $\tilde{V}^T_{t+1}(p,x)$ is nonnegative and concave in $p$ for all $x \in \Xset$. We show the same for $A^T_t(p, x)$ and $\tilde{V}^T_t(p, x)$. To that end, concavity of $\tilde{V}^T_{t+1}(p,x)$ implies the existence of a collection of affine functions 
such that
\begin{equation}
    \tilde{V}^T_{t+1}(p,x) = \inf_{z\in\mathcal{Z}(x)}\,\{a_z p + b_z\}
\end{equation}
for scalars $b_z, a_z$ for each $z \in \Zcal(x)$. 
Then, \eqref{eq:def.AT} gives
\begin{align}
\begin{aligned}
        A^T_t(p,x)
        &= \mathbb{E}\left[\tilde{V}^T_{t+1}(\Phi(X^+,x,p),x)|p,x,\pi_1(x)\right]\\
        &= \inf_{z\in\mathcal{Z}(x)}\mathbb{E}\left[a_z\Phi(X^+,x,p)+b_z|p,x,\pi_1(x)\right]\\
        &= \inf_{z\in\mathcal{Z}(x)}a_z\mathbb{E}\left[\Phi(X^+,x,p)|p,x,\pi_1(x)\right]+b_z.
\end{aligned}\label{eq:A_inf}
\end{align}
Leveraging \eqref{eq:phi_exp} in the above relation, we get
\begin{equation}
    A^T_t(p,x) = \inf_{z\in\mathcal{Z}(x)}\,\{a_z(p + \rho(1-p))+b_z\},
\end{equation}
proving that $A^T_t(\cdot,x)$ is concave. Nonnegativity of $\tilde{V}^T_{t+1}(p,x)$ yields the same for $A^T_t(p, x)$. Further,  from \eqref{eq:DP_suff_stage_cost}, $\tilde{V}^T_t(p,x)$ equals the pointwise minimum of a pair of concave and nonnegative functions of $p$, and therefore, is itself concave and nonnegative in $p\in[0,1]$. This completes the proof of step (a). 

\subsubsection*{$\bullet$ Step (b). Proving the threshold structure of the stopping rule in the $T$-period problem}
From \eqref{eq:DP_suff_stage_cost}, for $t=0,\ldots,T-1 $ $\tilde{V}^T_t(p, x)$ is the minimum of $\lambda(1-p)$ and $p + A^T_t(p, x)$. The first among these functions decreases monotonically from $\lambda$ to zero as $p$ sweeps from zero to unity. The second function is nonnegative and concave, owing to the same nature of $A^T_t(p, x)$. In the rest of this step, we prove that  
\begin{align}
    A^T_t(1, x) = 0
    \label{eq:ATt.1}
\end{align} 
for all $x\in\Xset$, and $t=0,\ldots, T-1$, which in turn yields that $p + A^T_t(p, x)$ can intersect $\lambda(1-p)$ at most once. If it does not intersect (equivalently $\lambda < A^T_t(0, x)$), then it is always optimal to stop at time $t$. Otherwise, the unique crossing point ${p}^\circ_t(x)$ satisfies
\begin{equation}
    \lambda(1-p^{\circ}_t(x))= {p}^\circ_t(x) + A^T_t({p}^\circ_t(x),x),
\end{equation}
implying that it is optimal to stop at time $t$, when $p$ at time $t$ exceeds ${p}^\circ_t(x)$. 

We prove \eqref{eq:ATt.1} by backward induction. From \eqref{eq:DP_suff_stage_cost_T}, we have $\tilde{V}^T_T(1,x) = 0$. Notice that \begin{align}
    p = 1 \implies \bar{p} = 1 \implies  \Phi(x^+,x, 1) = 1,
\end{align}
which in turn yields
    \begin{equation}
        \begin{split}
            A^T_{T-1}(1,x) 
            &= \E\left[\tilde{V}^T_T(\Phi(X^+,x, 1),X^+)|1,x,\pi_1(x)\right]
            \\
            &= \E\left[\tilde{V}^T_T(1,X^+)|1,x,\pi_1(x)\right] 
            \\
            &= 0,
        \end{split}
        \label{eq:ATt.1.2}
    \end{equation}
    proving the base case. 
    Per induction hypothesis, assume that $A^T_{t+1}(1, x) = 0$ for all $x\in\Xset$. Then, \eqref{eq:DP_suff_stage_cost} implies
    \begin{equation}
        \begin{split}
       \tilde{V}^T_{t+1}(1,x)=\min\{0, 1 + A^T_{t+1}(1,x)\} = 0.
        \end{split}
    \end{equation}
    Proceeding along the lines of \eqref{eq:ATt.1.2} completes the induction argument. This concludes the proof of the threshold structure of an optimal switching policy for the $T$-period problem.

\subsubsection*{$\bullet$ Step (c). Showing that $\tilde{v}_\star(p,x) = \lim_{T\to \infty}\tilde{V}^T_t(p,x)$}
For any stopping rule, extending the horizon length can only decrease the total cost, and hence, 
\begin{equation}\label{eq:V_monotonic}
    \tilde{V}^{T+1}_t(p,x)\leq \tilde{V}^T_t(p,x)
\end{equation}
for all $t=0,\ldots,T$, $p\in[0,1]$ and $x\in \Xset$.
Since $\tilde{V}^T_t$ is bounded below by zero for any finite $t$, the monotonically decreasing sequence converges as $T\to\infty$. Call this limit $\tilde{V}_t^\infty(p,x)$. Next, we show that this limit is independent of $t$.
The memoryless property of the geometric prior yields
\begin{equation}
    \tilde{V}^{T+1}_{t+1}(p,x) = \tilde{V}^T_t(p,x) 
    \implies    \tilde{V}^{\infty}_{t+1}(p,x) =\tilde{V}^{\infty}_{t}(p,x)
\end{equation}
for all $p\in[0,1]$ and $x \in \Xset$. The second equality follows from driving $T\to\infty$. Call this quantity $\tilde{V}(p,x)$. Next, we prove that $\Bcal \tilde{V} = \tilde{V}$.

From \eqref{eq:DP_suff_stage_cost},  $ \tilde{V}^T_t(p, x) \in [0, \lambda]$. Appealing to the dominated convergence theorem, \eqref{eq:def.AT} then allows us to infer
\begin{align}
    \hspace{-0.1in}\lim_{T\to\infty}A^{T}_t(p,x)= \E[\tilde{V}^{\infty}(\Phi(X^+,x,p),X^+)|p,x,\pi_1(x)]
\end{align} for all $p$ and $x$. Akin to the case with $\tilde{V}$, this limit does not depend on $t$; call it $A(p,x)$. Then, $\Bcal \tilde{V} = \tilde{V}$ follows from taking $T\to \infty$ on \eqref{eq:DP_suff_stage_cost}.

Now, \eqref{eq:v.lim}, the definition of $\tilde{V}^T_0$, and the first part of step (c) allows us to conclude
\begin{align}\label{eq:v_lim}
    \tilde{v}_\star(p,x) = \lim_{T\to\infty} \tilde{v}^T_\star(p,x) = \lim_{T\to\infty} \tilde{V}^T_0(p,x) = \tilde{V}(p,x),
\end{align}
implying that $\tilde{V}(p,x)$ characterizes the optimal cost of the infinite-horizon problem.

Nonnegativity and concavity of $\tilde{V}$ and $A$ follow from the same properties of $\tilde{V}^T$ and $A^T$, deduced in step (a). From step (b), we obtain $A(1,x) = 0$ for all $x \in \Xset$. Therefore, $\tilde{V}(1,x) = 0$ as well. Arguing along the lines of step (b), we infer that an optimal switching rule for the infinite horizon problem is threshold-type, where the threshold $p^\circ(x)$ satisfies
\begin{equation}
    \lambda(1-p^\circ(x))= {p}^\circ(x) + A({p}^\circ(x),x).
\end{equation}

Finally, we show that the switching time $\tau$ implied by $\tilde{v}_\star = \tilde{V}$ almost surely finite as follows.
\begin{subequations}\label{eq:boundedness_arg}
\begin{align}
    \E[\tau|p_0=p,X_0=x] &\leq  \E[|\tau-\Gamma||p_0=p,X_0=x] +\E[\Gamma|p_0=p,X_0=x] 
    \label{eq:tau_bound}
    \\
    &\leq \frac{1}{\min\{1,\lambda\}} \big(\E[(\tau-\Gamma)_+|p_0=p,X_0=x]
    \notag
    \\
    & \qquad \qquad + \lambda\E[(\Gamma-\tau)_+|p_0=p,X_0=x]\big) +\E[\Gamma|p_0=p,X_0=x]
    \label{eq:tau_bound2}
    \\
    & = \frac{1}{\min\{1,\lambda\}}\tilde{v}_\star(p, x) + \E[\Gamma|p_0=p,X_0=x]
    \label{eq:tau_bound3}
    \\
    & < \infty.
\end{align}
\end{subequations}
Therefore, $P\{\tau <\infty | p_0=p, X_0=x \} = 1$, proving the claim. Here, \eqref{eq:tau_bound} follows from triangle inequality, \eqref{eq:tau_bound2} from elementary calculations, and \eqref{eq:tau_bound3}
from the finiteness of $\tilde{v}_\star$ from \eqref{eq:v.lim}.

%% file: appendix_uniqueness.tex
\subsubsection*{$\bullet$ Step (d). Showing that $\Bcal$ has a unique fixed point in $\Ccal$}
Suppose $\Bcal G=G$ for some $G\in\Ccal$. Associate with $G$,
the following stopping rule
\begin{equation}
\begin{split}
    \tau_G = \min\{t\geq 0\,:\,\lambda(1-p_t)\leq p_t+\E[G(p_{t+1},X_{t+1})|p_t,x_t,\pi_1(x_t)]\}.
\end{split}
\end{equation}
Then, the definition of $\tau_G$ and $\Bcal G = G$ give
\begin{equation}
\begin{split}
    G(p,x) &= \E[(\tau_G-\Gamma))_+|p_0=p,X_0=x]+ \lambda P(\Gamma>\tau_G|p_0=p,X_0=x).
\end{split}
\end{equation}
Given that $G \in \Ccal$ is bounded, the argument in \eqref{eq:boundedness_arg} allows us to infer $P\{\tau_G<\infty|p_0=p,X_0=x\}=1$, implying that $\tau_G \in \Scal$. Then, $\tilde{v}_\star = \tilde{V}$ yields 
\begin{align}
    G(p,x)\geq \tilde{V}(p,x).
\end{align}
Next, we establish the reverse of the above inequality to conclude the proof of this step. It suffices to prove that $G(p,x) \leq \tilde{V}^T_{0}(p,x)$ for all $x\in\Xset$ and $p\in[0,1]$, which we show by backward induction. We begin by noticing that $G \in \Ccal$ implies
\begin{equation}
    G(p,x)\leq \lambda(1-p) = \tilde{V}^T_T(p,x) 
\end{equation}
for all $p\in[0,1]$, $x\in\Xset$, proving the base case. Assume the induction hypothesis $G(p,x) \leq \tilde{V}^T_{t+1}(p,x)$. Then, we have
\begin{align}
    \hspace{-.2in}\begin{aligned}
        \tilde{V}^T_{t}(p,x) 
        &= \min\{\lambda(1-p),1 + \E[\tilde{V}^T_{t+1}(\Phi(X^+, x, p), X^+)|p,x,\pi_1(x)]\}
        \\
        &\geq \min\{\lambda(1-p), 1 + \E[G(\Phi(X^+, x, p), X^+)|p,x,\pi_1(x)]\}
        \\
        &=G(p,x),
    \end{aligned}
\end{align}
where the last line uses $\Bcal G = G$.
This concludes the induction hypothesis and the proof of this step.

\subsubsection*{$\bullet$ Step (e). Proving  $\lim_{k \to \infty}\Bcal^k \psi = \tilde{V}$}
We utilize induction to establish that
\begin{align}
    0 \leq [\Bcal^{k+1}\psi](p,x) \leq [\Bcal^k \psi](p,x)
    \label{eq:Bk}
\end{align} 
for all $p\in [0,1]$, $x \in \Xset$, and $k\geq 0$. Then, the monotonically decreasing but bounded sequences $[\Bcal^k \psi](p,x)$ admits a limit. The rest follows from the observation that this limit must be a fixed point of $\Bcal$, for which $\tilde{V}$ is the unique candidate.

We prove \eqref{eq:Bk} via induction. Notice that 
\begin{align}
\hspace{-0.16in}\begin{aligned}
    [\Bcal\psi](p,x)
    = \min\{\psi(p,x), p + \E[\psi(\Phi(X^+,x,p),X^+)|p,x,\pi_1(x)]\}
    \leq \psi(p,x).
\end{aligned}
\end{align}
Here, $X^+$ denotes the random next state from $x$ under action $\pi_1(x)$.
Nonnegativity of $[\Bcal\psi](p,x)$ follows from the nonnegativity of both functions, whose minimum it equates to. Hence, \eqref{eq:Bk} holds for $k=0$. Next, assume that \eqref{eq:Bk} holds for $k$, and then we have
\begin{align}
\hspace{-0.1in}
    \begin{aligned}
        [\Bcal^{k+2}\psi](p,x)
        &=\min\{\lambda(1-p), p+\E\left[[\Bcal^{k+1}\psi](\Phi(X^+,x,p),X^+)|p,x,\pi_1(x)\right]\}
        \\
        &\overset{(a)}{\leq}\min\{\lambda(1-p), p+\E\left[[\Bcal^{k}\psi](\Phi(X^+,x,p),X^+)|p,x,\pi_1(x)\right]\}
        \\
        &= [\Bcal^{k+1}\psi](p,x),
    \end{aligned}
\end{align}
where (a) utilizes the induction hypothesis. This completes the proof of \eqref{eq:Bk} and Theorem \ref{thm:optimal_thresh}.

%% file: root_arxiv.bbl
\begin{thebibliography}{VBP93}

\bibitem[Ber12]{bertsekas2012dynamic}
Dimitri Bertsekas.
\newblock {\em Dynamic programming and optimal control: Volume I}, volume~1.
\newblock Athena scientific, 2012.

\bibitem[BLH17]{banerjee2017quickest}
Taposh Banerjee, Miao Liu, and Jonathan~P How.
\newblock Quickest change detection approach to optimal control in markov
  decision processes with model changes.
\newblock In {\em 2017 American control conference (ACC)}, pages 399--405.
  IEEE, 2017.

\bibitem[Lor71]{lorden1971procedures}
Gary Lorden.
\newblock Procedures for reacting to a change in distribution.
\newblock {\em The annals of mathematical statistics}, pages 1897--1908, 1971.

\bibitem[LP17]{levin2017markov}
David~A Levin and Yuval Peres.
\newblock {\em Markov chains and mixing times}, volume 107.
\newblock American Mathematical Soc., 2017.

\bibitem[Mon82]{monahan1982state}
George~E Monahan.
\newblock State of the art—a survey of partially observable markov decision
  processes: theory, models, and algorithms.
\newblock {\em Management science}, 28(1):1--16, 1982.

\bibitem[Pol85]{pollak1985optimal}
Moshe Pollak.
\newblock Optimal detection of a change in distribution.
\newblock {\em The Annals of Statistics}, pages 206--227, 1985.

\bibitem[PT12]{polunchenko2012state}
Aleksey~S Polunchenko and Alexander~G Tartakovsky.
\newblock State-of-the-art in sequential change-point detection.
\newblock {\em Methodology and computing in applied probability},
  14(3):649--684, 2012.

\bibitem[Put14]{puterman2014markov}
Martin~L Puterman.
\newblock {\em Markov decision processes: discrete stochastic dynamic
  programming}.
\newblock John Wiley \& Sons, 2014.

\bibitem[Rit90]{ritov1990decision}
Yaacov Ritov.
\newblock Decision theoretic optimality of the cusum procedure.
\newblock {\em The Annals of Statistics}, pages 1464--1469, 1990.

\bibitem[She25]{shewhart1925application}
Walter~A Shewhart.
\newblock The application of statistics as an aid in maintaining quality of a
  manufactured product.
\newblock {\em Journal of the American Statistical Association},
  20(152):546--548, 1925.

\bibitem[Shi07]{shiryaev2007optimal}
Albert~N Shiryaev.
\newblock {\em Optimal stopping rules}, volume~8.
\newblock Springer Science \& Business Media, 2007.

\bibitem[VB14]{veeravalli2014quickest}
Venugopal~V Veeravalli and Taposh Banerjee.
\newblock Quickest change detection.
\newblock In {\em Academic press library in signal processing}, volume~3, pages
  209--255. Elsevier, 2014.

\bibitem[VBP93]{veeravalli1993decentralized}
Venugopal~V Veeravalli, Tamer Basar, and H~Vincent Poor.
\newblock Decentralized sequential detection with a fusion center performing
  the sequential test.
\newblock {\em IEEE Transactions on Information Theory}, 39(2):433--442, 1993.

\end{thebibliography}
